\documentclass[11pt]{article}

\usepackage{amsthm}
\usepackage{amsmath}
\usepackage{amsfonts,amssymb}
\usepackage{graphicx,color}
\usepackage{boxedminipage}
\usepackage{framed}
\usepackage{mathtools}
\usepackage{enumitem}
\usepackage{thmtools}
\usepackage{thm-restate}
\usepackage{xspace}
\usepackage{todonotes}
\usepackage{footnote}
\usepackage{nicefrac}
\usepackage[numbers]{natbib}
\usepackage[a4paper,top=2cm,bottom=2cm,left=1in,right=1in,marginparwidth=1.75cm]{geometry}
\usepackage{appendix}
\usepackage{todonotes}
\usepackage{footnote}
\usepackage[pdftex, plainpages = false, pdfpagelabels, 
                 bookmarks=false,
                 bookmarksopen = true,
                 bookmarksnumbered = true,
                 breaklinks = true,
                 linktocpage,
                 colorlinks = true,  
                 linkcolor = blue,
                 urlcolor  = blue,
                 citecolor = red,
                 anchorcolor = green,
                 hyperindex = true,
                 hyperfigures
                 ]{hyperref} 
 \usepackage{xifthen}
 \usepackage{tabularx}
\usepackage{tikz} 
\usetikzlibrary{calc,math,patterns,shapes,backgrounds}
\tikzstyle{path} = [color=black,opacity=.30,line cap=round, line join=round, line width=10pt]

\usepackage[nameinlink]{cleveref}  

\usepackage{etoolbox}

\newlength{\RoundedBoxWidth}
\newsavebox{\GrayRoundedBox}
\newenvironment{GrayBox}[1]%
   {\setlength{\RoundedBoxWidth}{.93\textwidth}
    \def\boxheading{#1}
    \begin{lrbox}{\GrayRoundedBox}
       \begin{minipage}{\RoundedBoxWidth}}%
   {   \end{minipage}
    \end{lrbox}
    \begin{center}
    \begin{tikzpicture}%
       \node(Text)[draw=black!20,fill=white,rounded corners,%
             inner sep=2ex,text width=\RoundedBoxWidth]%
             {\usebox{\GrayRoundedBox}};
        \coordinate(x) at (current bounding box.north west);
        \node [draw=white,rectangle,inner sep=3pt,anchor=north west,fill=white] 
        at ($(x)+(6pt,.75em)$) {\boxheading};
    \end{tikzpicture}
    \end{center}}     

\newenvironment{defproblemx}[2][]{\noindent\ignorespaces%
                                \FrameSep=6pt%
                                \parindent=0pt%
                \vspace*{-1.5em}
                \ifthenelse{\isempty{#1}}{%
                  \begin{GrayBox}{\textsc{#2}}%
                }{%
                  \begin{GrayBox}{\textsc{#2}  parameterized by~{#1}}%
                }
                \begin{tabular*}{\textwidth}{@{\hspace{.1em}} >{\itshape} p{1.8cm} p{0.8\textwidth} @{}}%
            }{
                \end{tabular*}%
                \end{GrayBox}%
                \ignorespacesafterend
            }

\newcommand{\defproblema}[3]{
  \begin{defproblemx}{#1}
    Input:  & #2 \\
    Task: & #3
  \end{defproblemx}
}%

\newtheorem{lemma}{Lemma}

\newtheorem{proposition}{Proposition}

\newtheorem{observation}{Observation}

\crefname{@theorem}{Theorem}{Theorems}
\crefname{proposition}{Proposition}{Propositions}
\crefname{redrule}{Reduction Rule}{Reduction Rules}

\DeclareMathOperator{\operatorClassNP}{NP}
\newcommand{\classNP}{\ensuremath{\operatorClassNP}\xspace}

\DeclareMathOperator{\operatorClassCoNP}{coNP}
\newcommand{\classCoNP}{\ensuremath{\operatorClassCoNP}}
\DeclareMathOperator{\operatorClassFPT}{FPT\xspace}
\newcommand{\classFPT}{\ensuremath{\operatorClassFPT}\xspace}
\DeclareMathOperator{\operatorClassW}{W}
\newcommand{\classW}[1]{\ensuremath{\operatorClassW[#1]}}

\DeclareMathOperator{\operatorClassParaNP}{Para-NP\xspace}
\newcommand{\classParaNP}{\ensuremath{\operatorClassParaNP}\xspace}

\newcommand{\cqed}{\ensuremath{\lhd}}
\makeatletter
\newenvironment{claimproof}{\par
  \pushQED{\cqed}%
  \normalfont \topsep6\p@\@plus6\p@\relax
  \trivlist
  \item\relax
  {\itshape
    Proof of the claim\@addpunct{.}}\hspace\labelsep\ignorespaces
}{%
\hfill\popQED\endtrivlist\@endpefalse
}
\makeatother

\newcommand{\ftb}{\textsc{Fault-Tolerant Basis}\xspace}
\newcommand{\wftb}{\textsc{Weighted Fault-Tolerant Basis}\xspace}

\newcommand{\Oh}{\mathcal{O}}

\newcommand{\cM}{\mathcal{M}}
\newcommand{\cI}{\mathcal{I}}
\newcommand{\cB}{\mathcal{B}}
\newcommand{\cF}{\mathcal{F}}
\newcommand{\cZ}{\mathcal{Z}}

\newcommand{\rank}{\mathsf{rank}}
\newcommand{\cl}{\mathsf{cl}}

\newcommand{\bfA}{\mathbf{A}}

\title{Fault-Tolerant Matroid Bases\thanks{The research leading to these results have  been supported by the Research Council of Norway via the BWCA project (grant no. 314528), the Franco-Norwegian AURORA project (grant no. 349476), 
and the ERC Horizon 2020 research and innovation programme (grant agreement No. 819416).}}

\author{Matthias Bentert\thanks{University of Bergen, Norway. \texttt{\{Matthias.Bentert, Fedor.Fomin, Petr.Golovach\}@uib.no}} \and Fedor V. Fomin \addtocounter{footnote}{-1}\footnotemark \and Petr A. Golovach\addtocounter{footnote}{-1}\footnotemark \and 
Laure Morelle\thanks{LIRMM, Univ Montpellier, CNRS, Montpellier, France. \texttt{laure.morelle@lirmm.fr}}
}

\date{}

\begin{document}

\maketitle

\begin{abstract}
We investigate the problem of constructing \emph{fault-tolerant bases} in matroids. Given a matroid \(\cM\) and a redundancy parameter~\(k\), a \(k\)-fault-tolerant basis is a minimum-size set of elements such that, even after the removal of any \(k\) elements, the remaining subset still spans the entire ground set. Since matroids generalize linear independence across structures such as vector spaces, graphs, and set systems, this problem unifies and extends several fault-tolerant concepts appearing in prior research.

Our main contribution is a fixed-parameter tractable (\classFPT) algorithm for the \(k\)-fault-tolerant basis problem, parameterized by both \(k\) and the rank \(r\) of the matroid. This two-variable parameterization by \(k + r\) is shown to be tight in the following sense. On the one hand, the problem is already  \classNP-hard for \(k=1\). On the other hand, it is \classParaNP-hard for~\(r \geq 3\) and polynomial-time solvable for~$r \leq 2$.  
\end{abstract}

\section{Introduction}
 A basis in a $d$-dimensional vector space $\mathcal{V}$ is a set of exactly $d$ linearly independent vectors~$\{\mathbf{v}_1, \dots, \mathbf{v}_d\}$. Any vector $\mathbf{x} \in \mathcal{V}$ 
 has a \emph{unique} representation
$
  \mathbf{x} \;=\; \alpha_1 \mathbf{v}_1 \;+\;\cdots\;+\; \alpha_d \mathbf{v}_d.
$
However, if certain basis vectors are lost or corrupted---for instance, in a distributed storage system where vectors are stored across multiple servers that can malfunction---it may be challenging or impossible to accurately recover $\mathbf{x}$.
The algorithmic question we address in this paper is:
\emph{For an expected number~$k$ of failures, what is the minimum-size $k$-fault-tolerant set of vectors that still guarantees data reconstructability?}

A natural way to capture the notions of basis and independence is through \emph{matroids}, which generalize linear independence to a variety of combinatorial settings (like graphs, set systems, vector spaces).
Let $\cM = (E, \cI)$ be a matroid, and let $k$ be a nonnegative integer. We say that~$B \subseteq E$ is a \emph{$k$-fault-tolerant basis} of $\cM$ if $B$ is a minimum-size set such that, for every subset~$F \subseteq B$ of size at most $k$, we have $\rank(B \setminus F) = \rank(\cM)$. (All necessary matroid definitions appear in the Preliminaries section.)
Equivalently, this condition means that~${\cl(B \setminus F) = E}$, where $\cl(\cdot)$ denotes matroid closure. Note that this concept generalizes the standard basis definition: a $0$-fault-tolerant basis is simply a basis of $\cM$. However, while every matroid has a basis, it may not admit a $k$-fault-tolerant basis for $k \ge 1$.
We investigate the following problem.

\defproblema{\ftb}
{A matroid $\cM=(E,\cI)$ and  nonnegative integer $k$.}
{Find a $k$-fault-tolerant basis $B\subseteq E$ or correctly decide that none exists.}
Let us give some examples of \ftb. 

\subparagraph{Linear Matroids.}
A \emph{linear matroid} over a field $\mathbb{F}$ can be represented by a set of vectors~$\{\mathbf{v}_1, \dots, \mathbf{v}_m\}$ in a vector space $\mathcal{V} \subseteq \mathbb{F}^d$. The rank of the matroid corresponds to $\dim(\mathcal{V})$.
A basis here is simply a set of $d$ linearly independent vectors. 
A set of vectors $\{\mathbf{u}_1, \ldots, \mathbf{u}_r\}$  is $k$-fault-tolerant if, after removing up to $k$ vectors, the remaining set still spans the entire subspace $\mathcal{V}$.    Such concepts naturally arise in {distributed storage}---where vectors might be stored 
on different servers \cite{dimakis2010network}---and in  {coding theory}, where losing some symbols or vectors 
should not destroy the ability to recover the entire space~\cite{elad2010sparse}.

\subparagraph{Graphic Matroids.} For a connected graph~$G$, its \emph{graphic matroid}~$\mathcal{M}(G)$ has ground set~$E(G)$, and a set~$I\subseteq E(G)$ is independent if it forms a {forest}. A basis is thus a  {spanning tree} of $G$.
In this matroid, a set of edges $B \subseteq E(G)$ is $k$-fault-tolerant if, upon removing any~$k$~edges, the subgraph remains connected.  Thus, a $k$-fault-tolerant basis is a $(k+1)$-edge connected spanning subgraph with minimum number of edges. This problem is known in the literature as \textsc{$k$-Edge-Connected Spanning Subgraph}~\cite{ChalermsookHNSS22,Fernandes98,GabowGTW09}.
It is motivated by network applications as its task is to compute a minimum-cost sub-network that is resilient against up to $k$ link failures (see also~\cite{BentertSS23}).

\subparagraph{Gammoids.}
A \emph{gammoid} is derived from a directed (or undirected) graph $D$.  
Let $X$ and~$Y$ be two distinguished sets of vertices of $V(D)$.  
Within the set $X$, define a subset $U \subseteq X$ to be \emph{independent} if there exist $|U|$ vertex-disjoint paths in $G$ originating from vertices in~$Y$ and ending in the vertices of $U$.  The basis of the gammoid is the maximum set of vertices from~$X$ that can be reached by vertex-disjoint paths from~$Y$. In this case, a vertex set~$Z\subseteq X$ is a
 $k$-fault-tolerant basis if  upon removing any $k$ vertices from $Z$, the remaining vertices  are still linkable from sources in $Y$ without losing rank.  This problem is related to the problem of finding fault-tolerant disjoint paths studied by Adjiashvili et al.~\cite{AHMS22}.

\subparagraph{Transversal Matroids.}
A \emph{transversal matroid} $\mathcal{M}(U,\mathcal{S})$ is  described via a bipartite graph between a ground set $U$ and a set of ``target positions'' $\mathcal{S}$. A set $I \subseteq U$ is independent if it can be matched injectively to distinct elements in $\mathcal{S}$.
The basis of the transversal matroid is a maximum-size subset $I \subseteq U$ that can be perfectly matched into $\mathcal{S}$. 
Finding a $k$-fault-tolerant basis for a transversal matroids is naturally related to robust matching problems, where one wants to preserve a perfect (or maximum) matching under the loss of a few vertices or edges~\cite{dourado2015robust,hommelsheim2021secure}.


\subsection{Our Contribution}
We investigate the parameterized complexity of \ftb under natural parameterizations by $k$ and the rank $r$ of the input matroid. (We refer to the textbook by Cygan et al.~\cite{PCBookCygenetal} for an introduction to parameterized complexity.) Our main result is that the problem is \classFPT when parameterized by both $k$ and $r$.

\begin{restatable}{theorem}{fptrankandk}
\label{thm:fpt}
\ftb can be solved in $(kr)^{\Oh(kr^4)}\cdot n^{\Oh(1)}$ time on $n$-element matroids of rank at most $r$ given by independence oracles.
\end{restatable}
 
To prove \Cref{thm:fpt}, we design an algorithm that recursively decomposes the input matroid and identifies a small \emph{core set} containing any possible $k$-fault-tolerant basis, thereby restricting the problem to a bounded search space. More precisely, the algorithm locates a set $W$ of \emph{important elements} with bounded size such that, if the input matroid $\cM$ admits a $k$-fault-tolerant basis, then there is one contained entirely in $W$.
The construction of $W$ is guided by the observation (see~\Cref{obs:minbase}) that if $X$ is a rank-$r$ set with $r+k$ elements such that \emph{every}~$r$-element subset of $X$ is independent, then $X$ is already a $k$-fault-tolerant basis. Hence, finding such an $X$ would immediately solve the problem. Otherwise, we identify an inclusion-maximal rank-$r$ set $X$ (of bounded size) such that for every $r$-element subset~$Y \subseteq X$, the rank of $Y$ is $r$. 
The crucial insight here is that the closure $\cl(X)$ of $X$ can be expressed as the union of $\cl(S)$ over all $(r-1)$-element subsets $S$ of $X$. We then proceed recursively, selecting elements of $W$ by searching for analogous sets in the closures of each $S$, applying the same approach used to choose $X$.
Once we have computed~$W$, we determine whether a $k$-fault-tolerant basis exists within $W$ by enumerating all candidate subsets via a brute-force algorithm. If such a basis is found, it is also a valid solution for the original matroid.

We then show that the result of  \Cref{thm:fpt} is tight. First, by adapting  the results of Fomin et al.~\cite{FominGLS18} for our purposes, we observe the following lower bound.

\begin{restatable}{proposition}{hardtodecide}
\label{prop:hardtodecide}
It is \classW{1}-hard for the parameterization by $k$ to decide whether a given linear matroid has a $k$-fault-tolerant basis. 
\end{restatable}

For the parameterization of \ftb by $k$, recall that \ftb for graphic matroids is equivalent to finding a $(k+1)$-edge connected spanning subgraph with the minimum number of edges. As was shown by Fernandes~\cite{Fernandes98}, an $n$-vertex graph $G$ has a 2-edge connected spanning subgraph with at most $n$ edges if and only if $G$ has a Hamiltonian cycle. Thus, \ftb is intractable already for $k=1$.  Taking into account the inapproximability lower bounds for higher connectivities established by Gabow et al.~\cite{GabowGTW09}, we obtain the following observation. 

\begin{observation}\label{obs:hard-k}
   For every integer $k\geq 1$,
it is \classNP-hard to decide whether a graphic matroid~$\cM$ given with an integer $b\geq 1$ has a $k$-fault-tolerant basis of size at most $b$.
\end{observation}

For the parameterization by the rank, we establish a dichotomy---for any fixed $r\geq 3$, \ftb is intractable, but for $r\leq 2$, the problem can be solved in polynomial time. In fact, for $r\leq 2$, we solve 
the more general weighted variant of \ftb. In \wftb, we are given a matroid $\cM$ together with a weight function $w\colon E(\cM)\rightarrow \mathbb{Z}_{\geq 0}$, and the task is to find a set $B\subseteq E(\cM)$ of minimum total weight such that for any set $F\subseteq B$ of size at most $k$, $\rank(B\setminus F)=\rank(\cM)$. Our result is summarized in the following theorem.

\begin{restatable}{theorem}{hardness}
\label{thm:hard-r}
For any integer $r\geq 3$, it is \classNP-hard to decide whether a linear matroid $\cM$ of rank $r$ over rationals given together with integers $k$ and $b$ has a $k$-fault-tolerant basis of size at most $b$. For 
$r\leq 2$, \wftb can be solved in $\Oh(n^2)$ on matroids  
given by an independence oracle.
\end{restatable}
In particular, one can find a $k$-fault-tolerant basis of vectors in $\mathbb{R}^2$ in polynomial time but the problem becomes \classNP-hard  in $\mathbb{R}^3$.
To obtain the hardness for $r\geq 3$, we use the result of Froese et al.~\cite{FroeseKNN17} stating that, given a set~$P$ of points on the plane and a positive integer $k$, it is NP-hard to decide whether $P$ contains at least~$k$ points in general position. To establish the claim for $r\leq 2$, it is convenient to show 
a more general result for partition matroids which may be of independent interest.

\begin{restatable}{proposition}{partition}
\label{thm:poly-part}
Given an $n$-element partition matroid~$\cM$ with unit capacities together with a weight function $w\colon E(\cM)\rightarrow\mathbb{Z}_{\geq 0}$ and integers $r,k\geq 0$, 
one can find in~$\Oh(n^2)$~time either a set $X\subseteq E(\cM)$ of minimum weight  
such that, for any set $F\subseteq X$ of size at most $k$, it holds that~$\rank(X\setminus F)\geq r$ or correctly decide that such a set does not exist.   
\end{restatable}

\subsection{Related Work}
Adjiashvili, Stiller, and Zenklusen~\cite{Adjiashvili2015bulk} introduced a 
\textit{Bulk-Robustness} model of combinatorial optimization. They studied several 
instances of this framework, including the \textsc{Bulk-Robust Minimum Matroid Basis} 
problem, which is most relevant to our work.
In this setting, for a matroid $\cM = (E, \cI)$, we are given a collection of 
\emph{interdiction sets}
$
  \Omega \;=\; \{F_1, F_2, \ldots, F_m\},
$
where~$F_i \subseteq E(\cM)$ for each~$i \in \{1,\ldots,m\}$. The goal is to find a cheapest subset $X \subseteq  E(\cM)$
such that $X \setminus F_i$ contains a basis of $\cM$ for every 
$i \in \{1, \dots, m\}$.
Adjiashvili, Stiller, and Zenklusen provided an $\Oh(\log m + \log r)$-approximation
algorithm for \textsc{Bulk-Robust Minimum Matroid Basis}, where $r$ is the rank of the matroid.
The problem  {\ftb} can be viewed as a special case of \textsc{Bulk-Robust Minimum Matroid Basis}
when $\Omega$ consists of all subsets of~$E$ of size $k$.

There are several prior works investigating fault-tolerance for classic optimization problems. 
The model of \textsc{$s$-$t$-Path} and \textsc{$s$-$t$-Flow} problems with safe and vulnerable edges was introduced by Adjiashvili et al.~\cite{AHMS22}, who studied the approximability of these problems. Subsequently, generalizations were also studied \cite{BCGI22,BCHI21,CJ22}. Bentert et al.~\cite{BentertSS23}  studied the parameterized complexity of computing a fault-tolerant spanning tree. Approximation algorithms and inapproximability lower bounds for 
\textsc{$k$-Edge-Connected Spanning Subgraph} have been considered in~\cite{ChalermsookHNSS22,Fernandes98,GabowGTW09,KhullerV92}. 

More generally, \textsc{\ftb} belongs to \emph{robust optimization}, a branch of optimization
that adapts classic optimization-theoretic tools to settings with uncertainty.
For a comprehensive introduction, see the survey by Bertsimas et al.~\cite{Bertsimas2011robust}.

\section{Preliminaries}\label{sec:prelim}
We refer to the book of Oxley~\cite{Oxley92} for a detailed
introduction to matroids. A pair $\cM=(E,\cI)$, where $E$ is a finite
\emph{ground} set and $\cI$ is a family of subsets of the ground set,
called \emph{independent sets} of $\cM$, is a \emph{matroid} if it satisfies the
following conditions, called \emph{independence axioms}:
\begin{description}
\item[(I1)]  $\emptyset \in \cI$. 
\item[(I2)]  If $A\subseteq B\subseteq E$ and $B\in \cI$ then $A\in\cI$. 
\item[(I3)] If $A, B  \in \cI$  and $ |A| < |B| $, then there is $ e \in  B \setminus A$  such that $A\cup\{e\} \in \cI$.
\end{description}
We use $E(\cM)$ and $\mathcal{I}(\cM)$ to denote the ground set and the set of
independent sets, respectively. 
An inclusion-maximal independent set $B$ is called a
\emph{basis} of $\cM$. We use $\cB(\cM)$ to denote the set of bases of $\cM$.
All the bases of $\cM$ have the same size called the \emph{rank} of $\cM$ and 
denoted by~$\rank(\cM)$. The \emph{rank} of a subset $A\subseteq E(\cM)$, denoted by~$\rank(A)$, is the maximum size of an independent set $X\subseteq A$; the
function $\rank\colon 2^{E(\cM)}\rightarrow \mathbb{Z}$ is the \emph{rank}
function of~$M$. A set~$A\subseteq E(\cM)$ \emph{spans} an element $x\in E(\cM)$ if
$\rank(A\cup\{x\})=\rank(A)$. The \emph{closure} (or \emph{span}) of $A$ is the
set $\cl(A)=\{x\in E(M)\mid A\text{ spans }x\}$. Closures satisfy the following
properties, called \emph{closure axioms}:
\begin{description}
\item[(CL1)]  For every $A\subseteq E(\cM)$, $A\subseteq \cl(A)$. 
\item[(CL2)]  If $A\subseteq B\subseteq E(\cM)$, then $\cl(A)\subseteq \cl(B)$.
\item[(CL3)]  For every $A\subseteq E(\cM)$, $\cl(A)=\cl(\cl(A))$.
\item[(CL4)]  For every $A\subseteq E(\cM)$, every $x\in E(\cM) \setminus A$, and every~$y\in\cl(A\cup\{x\})\setminus \cl(A)$, $x\in \cl(A\cup\{y\})$. 
\end{description}

For $S\subseteq E(M)$, the matroid obtained by \emph{deleting} $S$, denoted by $\cM'=\cM-S$, is the matroid such that
$E(\cM')=E(\cM)\setminus S$ and $\cI(\cM')=\{X\in \cI(\cM)\mid
S\cap X=\emptyset\}$.
Let $r\geq 0$ be an integer. The \emph{$r$-truncation} of a matroid $\cM$ is the matroid $\cM'$ with $E(\cM')=E(\cM)$ such that $X\in \cI(\cM')$ if and only if $X\in \cI(\cM)$ and $|X|\leq r$.
A matroid $\cM$ is a \emph{partition} matroid if there is a partition~$(P_1,\ldots,P_d)$ of the ground set and a $d$-tuple $(c_1,\ldots,c_d)$ of positive integers, called capacities, such that 
a set $X\subseteq E(G)$ is independent if and only if~$|X\cap P_i|\leq c_i$ for every~$i\in\{1,\ldots,d\}$. 
A $d\times n$-matrix
$\bfA$ over a field~$\mathbb{F}$ is a \emph{representation of a matroid $\cM$ over~$\mathbb{F}$} if
there is a one-to-one correspondence $f$ between $E(\cM)$ and the columns of~$\bfA$ such that for any~$X\subseteq E(\cM)$, it holds that~$X\in \cI(\cM)$ if and only if
the columns in~$\{f(x) \mid x \in X\}$ are linearly independent (as vectors of $\mathbb{F}^d$). If
$\cM$ has such a representation, then~$\cM$ has a
\emph{representation over $\mathbb{F}$}. A matroid $\cM$ admitting a representation over $\mathbb{F}$ is said to be \emph{linear}; a matroid is \emph{binary} if it has a representation over ${\sf GF}[2]$.   

In our algorithms for general matroids,  we assume
that input matroids are given by \emph{independence oracles}. Such an oracle for a matroid $\cM$, takes as input a set $X\subseteq E(\cM)$ and correctly answers in a constant time whether $X\in \cI(\cM)$. We note that some matroids could be given explicitly, for example, by their representations. 

\section{Basic Observations}
In this section, we prove a couple of simple results that will be helpful in the later discussion.
Recall that for an integer $k\geq 0$, $B\subseteq E(M)$ is a $k$-fault-tolerant basis of a matroid $\cM$ if $B$ is a subset of minimum cardinality such that $\rank(B\setminus F)=\rank(\cM)$ for every subset~$F\subseteq B$ of size at most $k$. We use the following bounds on the size of a $k$-fault-tolerant basis.

\begin{proposition}\label{prop:bounds}
Let $\cM$ be a matroid with $\rank(\cM)\geq 1$ and let $k\geq 0$ be an integer.  Then for a $k$-fault-tolerant basis $B$ of $\cM$, it holds that
\begin{equation*}
\rank(\cM)+k\leq |B|\leq (k+1)\rank(\cM).    
\end{equation*}
\end{proposition}

\begin{proof}
The lower bound immediately follows from the definition of a $k$-fault-tolerant basis.  
To show the upper bound, assume for the sake of contradiction that~$|B|>(k+1)\rank(\cM)$. 
We iteratively construct sets~$X_0,\ldots,X_k$ where~$X_0\subseteq B$ is an inclusion-maximal independent set in~$B$, and for each~$i\in\{1,\ldots,k\}$, $X_i$~is an inclusion-maximal independent set in~${B\setminus\big(\bigcup_{j=0}^{i-1} X_j\big)}$. 
As~$\rank(B)=\rank(\cM)$ and $|B|>(k+1)\rank(\cM)$, such sets~$X_0,\ldots,X_k$ exist. 
Set~${B'=\bigcup_{i=0}^kX_i}$. 
Consider an arbitrary $F\subseteq B$ of size at most $k$. 
By the pigeonhole principle, there is an index~${i\in\{0,\ldots,k\}}$ such that $X_i\cap F=\emptyset$. By construction, $B\setminus B'\subseteq B\setminus\big(\bigcup_{j=0}^{i-1} X_j\big)\subseteq \cl(X_i)$. 
Hence, $B\setminus F\subseteq (B'\setminus F)\cup (B\setminus B')\subseteq \cl(B'\setminus F)$. Because $B$ is a $k$-fault-tolerant basis, $\cl(B\setminus F)=E(\cM)$. Thus, 
$E(\cM)=\cl(B\setminus F)\subseteq \cl(\cl(B'\setminus F))=\cl(B'\setminus F)$. 
Since $F$ was chosen arbitrarily, we obtain that $B'$ is a $k$-fault-tolerant basis of $\cM$. 
However, $|B'|<|B|$ contradicting the choice of $B$.
\end{proof}

We next argue that the bounds in \Cref{prop:bounds} are tight.
First, note that the lower bound is tight as the matroid $\cI(\cM)=\{X\subseteq E(\cM)\colon |X|\leq k\}$ proves. 
We next argue that the upper bound is also tight. 
Let $e_1,\ldots,e_r$ be a basis of $\mathbb{R}^r$. Let $k$ and $n$ be integers such that~$0\leq k<n$. Consider the linear matroid $\cM$ with~${E(\cM)=\{je_i\mid 1\leq i\leq r\text{  and }1\leq j\leq n\}}$. It is easy to see that any $k$-fault-tolerant basis of~$\cM$ has to contain at least $k+1$ vectors of~$\{je_i\mid 1\leq j\leq n\}$ for each $i\in\{1,\ldots,r\}$. Thus, the size of a $k$-fault-tolerant basis is at least $(k+1)r$. 
From the other side, we have that, for~$B=\{je_i\mid 1\leq i\leq r\text{ and }1\leq j\leq k+1\}$, 
$\rank(B\setminus F)=r$ for any~$F\subseteq E(\cM)$ of size at most $k$. Thus, $B$ is
a $k$-fault-tolerant basis of size~$(k+1)r$ of~$\cM$.

Let $\cM$ be a matroid.
For a positive integer $h$, we say that a set $X\subseteq E(\cM)$ is \emph{$h$-uniform} if~$\rank(X)= h $ and~$\rank(Y)=h$ for every subset $Y\subseteq X$ of size $h$. The definition and \Cref{prop:bounds} yield the following.

\begin{observation}\label{obs:minbase}
Let $\cM$ be a matroid of rank $r\geq 1$, $k\geq 0$ be an integer, and~$B\subseteq E(\cM)$ be a set of size $k+r$. Then, $B$ is a $k$-fault-tolerant basis of $\cM$ if and only if $B$ is $r$-uniform.
\end{observation}

\begin{proof}
If $B$ is a $k$-fault-tolerant basis, then${\rank(B\setminus F)=r}$ holds for every set~$F\subseteq B$ of size at most $k$.
Since $|B|=k+r$, this implies that, for every set $X\subseteq B$ of size $r$, ${\rank(X)=\rank(B)=r}$, that is, $B$ is $r$-uniform. 
For the opposite direction, if $B$ is $r$-uniform, then, for any $X\subseteq B$ of size~$r$, $\rank(X)=\rank(B)=\rank(\cM)$.
Since~$|B|=k+r$, it holds for any~$F\subseteq B$ of size at most~$k$ that~$\rank(B\setminus F)=\rank(\cM)$. This completes the proof.     
\end{proof}

We conclude this section by proving \Cref{prop:hardtodecide} using the results of Fomin et al.~\cite{FominGLS18}. In particular, they studied \textsc{Rank $h$-Reduction}.
Here, the input is a binary matroid $\cM$ given by its representation over ${\sf GF}[2]$ and two positive integers $h$ and $k$.
The task is to decide whether there is a set $X\subseteq E(\cM)$ of size at most $k$ such that $\rank(\cM)-\rank(\cM-X)\geq h$.

\hardtodecide*

\begin{proof}
We reduce from \textsc{Rank $h$-Reduction} parameterized by~$k$, which is known to be W[1]-hard~\cite{FominGLS18}.
Let $(\cM,h,k)$ be an instance of \textsc{Rank $h$-Reduction} where $\cM$ is a binary matroid of rank $r$. We define $\cM'$ to be the $(r-h+1)$-truncation of $\cM$. Notice that $\cM'$ is a linear matroid and its representation (over a different field) can be constructed in  polynomial time by the result of Lokshtanov et al.~\cite{LokshtanovMPS18}. We claim that $\cM'$ has no $k$-fault-tolerant basis if and only if $(\cM,h,k)$ is a yes-instance of  \textsc{Rank $h$-Reduction}. 

To this end, note that $\cM'$ has no $k$-fault-tolerant basis if and only if there is a set~${X\subseteq E(\cM')}$ of size at most $k$ such that
$\rank(\cM'-X)<\rank(\cM')=\rank(\cM)-h+1$. Note that since~${\rank(\cM'-X)<\rank(\cM)-h+1}$ if and only if $\rank(\cM)-\rank(\cM-X)\geq h$, 
$\cM'$~has no \mbox{$k$-fault-tolerant} basis if and only if there is a set~$X\subseteq E(\cM')$ of size at most $k$ such that~${\rank(\cM)-\rank(\cM-X)\geq h}$. This concludes the proof.   
\end{proof}

We remark that since the hardness for \textsc{Rank $h$-Reduction} was proven by Fomin et al.~\cite{FominGLS18} via a polynomial-time reduction from~\textsc{Clique}, it is \classCoNP-hard to decide whether a linear matroid has a $k$-fault-tolerant basis.
We also note that the problem is in XP since we can decide in $n^{\Oh(k)}$ time whether an $n$-element matroid $\cM$ has a $k$-fault-tolerant basis---simply check for each subset $X\subseteq E(\cM)$ of size $k$ whether~${\rank(\cM-X)=\rank(\cM)}$.

\section{An FPT algorithm for the parameterization by rank and~$k$}\label{sec:fpt}
In this section, we prove \Cref{thm:fpt}. We construct a recursive branching algorithm that finds a set~$W$ of important elements of bounded size with the property that, if the input matroid~$\cM$ has a $k$-fault-tolerant basis, then there is a $k$-fault-tolerant basis $B\subseteq W$.
Note that a $k$-fault-tolerant basis has minimum size by definition.
Then, we select a $k$-fault-tolerant basis (if it exists) in $W$ using brute force.  The construction of $W$ is inspired by \Cref{obs:minbase}, which indicates that $h$-uniform sets are preferable in the construction of $k$-fault-tolerant bases. 
The following lemma is crucial for constructing $W$.

\begin{restatable}{lemma}{uniform}
\label{lem:uniform}
Let $\cM$ be a matroid of rank $r\geq 1$ and $k\geq 0$ be an integer. Let~$X\subseteq E(\cM)$ be an~$h$-uniform set of size at least $(h-1)[(k+1)r]^{r-1}+(k+1)r$ for some $1\leq h\leq r$. Then, for any $k$-fault-tolerant basis~$B$ of~$\cM$, there is a $k$-fault-tolerant basis $B'$ such that 
\begin{itemize}
\item[(i)] $B\setminus\cl(X)=B'\setminus\cl(X)$ and
\item[(ii)] $B'\cap\cl(X)\subseteq X$.
\end{itemize}
\end{restatable}

\begin{proof}
 Let $B$ be a $k$-fault-tolerant basis of $\cM$. If $B\subseteq \cl(X)$, then the claim is straightforward because in this case, $\cl(X)=E(\cM)$ and $r=h$.
 We can therefore select $B'$ to be the set of $k+r$ arbitrary elements of $X$ by~\Cref{obs:minbase}. We hence assume from now on that~$B\setminus\cl(X)\neq\emptyset$.
 
 Consider a $k$-fault-tolerant basis $B'$ satisfying (i), that is, $B\setminus\cl(X)=B'\setminus\cl(X)$, such that the size of $B'\cap (\cl(X)\setminus X)$ is minimum. We claim that $B'$ satisfies (ii), that is, ${B'\cap\cl(X)\subseteq X}$. 
Assume towards a contradiction that $B'\cap (\cl(X)\setminus X)\neq\emptyset$ and  there is an~${x\in B'\cap (\cl(X)\setminus X)}$. 
Let $Y=B'\setminus\{x\}$. 
Note that $Y$ is not a $k$-fault-tolerant basis of $\cM$.
Hence, there is a set~$F\subseteq Y$ of size at most $k$ such that $\rank(Y\setminus F)<r$. 
Notice that~$\rank(B'\setminus F)=r$ implies~$\rank(Y\setminus F)=r-1$. Denote by $\cF$ the family of all sets $F\subseteq Y$ of size at most~$k$ such that~$\rank(Y\setminus F)=r-1$. 
For each~$F\in \cF$, let~$Z_F\subseteq Y$ be an inclusion-maximal independent set in~$Y\setminus F$ and let~$\cZ=\{Z_F\mid F\in \cF\}$. Note that~$\cl(Z_F)=\cl(Y\setminus F)$ and therefore~$\rank(Z_F)=r-1$ for each~$F \in \cF$.
By~\Cref{prop:bounds}, it holds that~$|Y|\leq (k+1)r-1$. Hence, 
$|\cZ|\leq\binom{(k+1)r-1}{r-1}\leq[(k+1)r]^{r-1}$.
We next show that~$|\cl(Z_F)\cap X|\leq h-1$ for each~$F \in \cF$. 

To this end, assume towards a contradiction that $|\cl(Z_F)\cap X|\geq h$. Since~$X$ is $h$-uniform, we obtain that $X\subseteq \cl(\cl(Z_F)\cap X)$.
Then $\cl(X)\subseteq \cl(\cl(Z_F)\cap X)\subseteq \cl(Z_F)$ and, in particular, $x\in \cl(Z_F)=\cl(Y\setminus F)$. 
However, in this case, 
$B'\setminus F=(Y\cup\{x\})\setminus F\subseteq \cl(Y\setminus F)$. Since~$B'$ is a $k$-fault-tolerant basis, we have that 
$\rank(Y\setminus F)=\rank(B'\setminus F)=r$ contradicting~$F\in\cF$.

Recall that $|X|\geq (h-1)[(k+1)r]^{r-1}+(k+1)r$ and $B'$ has at least one element outside~$X$. Then, $|B'| \leq (k+1)r$ and therefore~$|X\setminus B'|\geq (h-1)[(k+1)r]^{r-1}+1$ by~\Cref{prop:bounds}.
Since~$|\cZ|<[(k+1)r]^{r-1}$ and, for each~$F\in \cF$, $|\cl(Z_F)\cap X|\leq h-1$, a simple counting argument shows that there exists an element~$y\in X\setminus B'$ such that $y\notin\cl(Z_F)$ for all~$F\in\cF$. Let~$B^*=Y\cup\{y\}=(B'\setminus\{x\})\cup\{y\}$. Note that by definition, (i)~$B^*\setminus \cl(X)=B'\setminus \cl(X)$, (ii)~$|B^*\cap(\cl(X)\setminus X)|< |B'\cap(\cl(X)\setminus X)|$, and (iii)~$|B^*|=|B'|$.
We next show that $B^*$ is a $k$-fault-tolerant basis.
Afterwards, we will show that this contradicts our choice of~$B'$.

Towards the former, let $F\subseteq B^*$ be of size at most $k$. We prove that $\rank(B^*\setminus F)=r$.
If $\rank(Y\setminus F)=r$, then $\rank(B^*\setminus F)\geq \rank(Y\setminus F)=r$. 
So we assume from now on that $\rank(Y\setminus F)<r$. Then, $\rank(Y\setminus F)=r-1$ as shown above. 
Assume towards a contradiction that $y\in F$. Since~$y\notin B'$, we have that 
$F'=(F\setminus \{y\})\cup\{x\}$ is a subset of $B'$ of size at most $k$. Since~$B'$ is a $k$-fault-tolerant basis, $\rank(B'\setminus F')=r$. However, $B'\setminus F'=Y\setminus F$ and $\rank(Y\setminus F)=r$, contradicting our assumption that~$\rank(Y \setminus F) < r$.
Hence, $y\notin F$ and therefore~$F\subseteq Y$. Recall that $\rank(Y\setminus F)=r-1$. This implies that $F\in \cF$ and $Z_F\in\cZ$. Since~$y\notin \cl(Z_F)$ by the choice of~$y$, 
$\rank(B^*\setminus F)=\rank((Y\setminus F)\cup\{y\})>\rank(Y\setminus F)=r-1$. Thus, $\rank(B^*\setminus F)=r$. 

Since~$\rank(B^*\setminus F)=r$ for every $F\subseteq B^*$ of size at most $k$ and $|B^*|=|B'|$, we have that~$B^*$ is a $k$-fault-tolerant basis. However, we also obtain that
(i)~$B^*\setminus \cl(X)=B'\setminus \cl(X)$ and (ii)~$|B^*\cap(\cl(X)\setminus X)|< |B'\cap(\cl(X)\setminus X)|$ contradicting the choice of $B'$. This 
shows that~$B'\cap\cl(X)\subseteq X$ and 
completes the proof.
\end{proof}

Now we show how to construct a set~$W$ of important elements of bounded size.

\begin{lemma}\label{lem:important}
There is an algorithm that, given a matroid $\cM$ of rank $r\geq 1$ with~${\{e\}\in \cI(\cM)}$ for each~$e\in E(\cM)$
and an integer $k\geq 0$, outputs a set $W\subseteq E(\cM)$ of at most~${r^{r^2}\cdot [(k+1)r]^{r^3}}$ elements in~$(rk)^{\Oh(r^2)}\cdot n^{\Oh(1)}$ time such that $\cM$ has a $k$-fault-tolerant basis if and only if~$\cM$ has a $k$-fault-tolerant basis $B\subseteq W$.
\end{lemma}

\begin{proof}
Let~$\cM$ be a matroid of rank~$r \geq 1$ such that~$\{e\}\in \cI(\cM)$ for any $e\in E(\cM)$ and let~$k$ be a non-negative integer.
We construct a recursive branching algorithm $\textsc{Important}(\cM,X)$ that takes as input a matroid $\cM$ and a non-empty independent set $X\in \cI(\cM)$, and outputs a set~$Y\subseteq \cl(X)$ of size at most 
$|X|^{|X|^2}\cdot [(k+1)r]^{r|X|^2}$
with the property that
for any $k$-fault-tolerant basis $B$ of $\cM$, there is a $k$-fault-tolerant basis $B'$ such that 
\begin{itemize}
\item[(i)] $B\setminus\cl(X)=B'\setminus\cl(X)$ and
\item[(ii)] $B'\cap\cl(X)\subseteq Y$.
\end{itemize}
Let $h=|X|$. The base case is $h=1$, where we do the following:
\begin{itemize}
\item Compute $\cl(X)$.
\item If $|\cl(X)|< (k+1)r$, then set $Y:=\cl(X)$ and output it.
\item Otherwise, define $Y$ to be the set of $(k+1)r$ arbitrary elements of $\cl(X)$ and output it.
\end{itemize}
For $h>1$, we do the following:
\begin{itemize}
\item Compute $\cl(X)$.
\item If $|\cl(X)|\leq (h-1)[(k+1)r]^{r-1}+(k+1)r$, then set $Y:=\cl(X)$, output it, and stop.
\item Find a $h$-uniform set $Z\subseteq\cl(X)$ that either has size $(h-1)[(k+1)r]^{r-1}+(k+1)r$ or is an inclusion-maximal $h$-uniform set of size at most
$(h-1)[(k+1)r]^{r-1}+(k+1)r-1$.
\item If $|Z|=(h-1)[(k+1)r]^{r-1}+(k+1)r$, then set $Y:=Z$ and output it.
\item If $Z$ is an inclusion-maximal $h$-uniform set of size at most
$(h-1)[(k+1)r]^{r-1}+(k+1)r-1$, then output~$Y:=\bigcup_{S\subseteq Z\text{ s.t. }|S|=h-1}Y_S$, where $Y_S$ is the output of $\textsc{Important}(\cM,S)$.
\end{itemize}
To compute an $h$-uniform set $Z$, we apply the following greedy procedure:
\begin{itemize}
\item Initially, set $Z:=X$.
\item While $Z\leq (h-1)[(k+1)r]^{r-1}+(k+1)r-1$, do the following: 
\begin{itemize}
\item For every $x\in\cl(X)\setminus Z$, check whether $Z\cup\{x\}$ is $h$-uniform and set $Z:=Z\cup\{x\}$ if this holds.
\item If $Z\cup\{x\}$ is not $h$-uniform for all $x\in \cl(X)\setminus Z$, then output $Z$ and stop. 
\end{itemize}
\item If $Z=(h-1)[(k+1)r]^{r-1}+(k+1)r$, then output $Z$.
\end{itemize}
The crucial property of the algorithm is given in the following claim. 

\begin{restatable}{claim}{correctness}
\label{cl:main}
$\textsc{Important}(\cM,X)$ outputs a set $Y\subseteq \cl(X)$ with the property that
for any $k$-fault-tolerant basis $B$ of $\cM$, there is a $k$-fault-tolerant basis $B'$ such that 
\begin{itemize}
\item[(i)] $B\setminus\cl(X)=B'\setminus\cl(X)$ and
\item[(ii)] $B'\cap\cl(X)\subseteq Y$.
\end{itemize}
\end{restatable}

\begin{claimproof}
Observe first that the algorithm is finite and always outputs some set $Y$ because the recursion stops when $h=1$, and in the step for $h>1$, we may recursively call the algorithm only for independent sets $S$ of size $h-1$.
Moreover, it is easy to verify that the described greedy strategy correctly constructs an $h$-uniform set $Z$.
Finally, we use induction on~$h$ to show that there is a fault-tolerant basis $B'$ satisfying~(i) and~(ii).

For the base case $h=1$, observe that because $\{e\}\in \cI(\cM)$ for any $e\in E(\cM)$ by the assumptions of the lemma, any non-empty subset of $\cl(X)$ is $1$-uniform. Then \Cref{lem:uniform} immediately implies the claim. 

For~$h\geq 2$, assume that the claim holds for all~$h' < h$.
If $|\cl(X)|\leq (h-1)[(k+1)r]^{r-1}+(k+1)r$, then the existence of $B'$ satisfying (i) and (ii) is trivial as~$Y = \cl(X)$ and hence, we can take~$B'=B$. If $Z$ is an $h$-uniform set of size $(h-1)[(k+1)r]^{r-1}+(k+1)r$, then $B'$ satisfying (i) and (ii) exists by~\Cref{lem:uniform}. Suppose that $Z$ is an inclusion-maximal $h$-uniform subset of $\cl(X)$. Then for any $x\in\cl(X)\setminus Z$, there is a set $S\subseteq Z$ of size $h-1$ such that $x\in\cl(S)$. Otherwise, we would include~$x$ in~$Z$. 
Hence, we have that 
$\cl(X)=\bigcup_{S\subseteq Z\text{ s.t. }|S|=h-1}\cl(Y_S)$, where $Y_S$ is the set obtained in the recursive call $\textsc{Important}(\cM,S)$. 

In our algorithm, we call $\textsc{Important}(\cM,S)$ for every $S\subseteq Z$ of size $h-1$. Denote these sets by~$S_1,\ldots,S_\ell$ and assume that the sets are indexed according to the order in which the algorithm is called for them.
We assume that $Y_i=Y_{S_i}$ for $i\in \{1,\ldots,\ell\}$. The algorithm sets~$Y=\bigcup_{i=1}^\ell Y_i$.

Consider arbitrary $k$-fault-tolerant basis $B=B_0$ of $\cM$. By the induction hypothesis, there a sequence $B_1,\ldots,B_\ell$ of $k$-fault-tolerant bases such that for every $i\in\{1,\ldots,\ell\}$, it holds that
\begin{itemize}
\item[(i$'$)] $B_{i-1}\setminus\cl(S_i)=B_i\setminus\cl(S_i)$ and
\item[(ii$'$)] $B_i\cap\cl(S_i)\subseteq Y_i$.
\end{itemize}
We define $B'=B_\ell$ and claim that (i) and (ii) hold for this choice of $B'$. 

To see (i), we show inductively that~$B_{j-1}\setminus \bigcup_{i=j}^\ell\cl(S_i)=B'\setminus \bigcup_{i=j}^\ell\cl(S_i)$ for each~${j\in\{1,\ldots,\ell\}}$. This holds for $j=\ell$, by (i$'$). Consider $j<\ell$. Then by (i$'$),
$B_{j-1}\setminus\cl(S_j)=B_j\setminus\cl(S_j)$. Also, by the inductive assumption, $B_{j}\setminus \bigcup_{i=j+1}^\ell\cl(S_i)=B'\setminus \bigcup_{i=j+1}^\ell\cl(S_i)$. Then 
\begin{align*}
B_{j-1}\setminus \bigcup_{i=j}^\ell\cl(S_i)=&\ (B_{j-1}\setminus \cl(S_j))\setminus \bigcup_{i=j+1}^\ell\cl(S_i)= (B_{j}\setminus \cl(S_j))\setminus \bigcup_{i=j+1}^\ell\cl(S_i)\\
=&\ (B_j\setminus \bigcup_{i=j+1}^\ell\cl(S_i))\setminus \cl(S_j)=(B'\setminus \bigcup_{i=j+1}^\ell\cl(S_i))\setminus \cl(S_j)\\
=&\ B'\setminus \bigcup_{i=j}^\ell\cl(S_i)
\end{align*}
proving the claim. Because $B=B_0$ and $\cl(X)=\bigcup_{i=1}^\ell(\cl(S_i))$,
we obtain that 
\begin{equation*}
B\setminus\cl(X)=B_0\setminus \bigcup_{i=1}^\ell\cl(S_i)=
B'\setminus \bigcup_{i=1}^\ell\cl(S_i)
=B'\setminus\cl(X).    
\end{equation*}
This proves (i).

To establish (ii), we inductively prove that~${B_j\cap\bigcup_{i=1}^j\cl(S_i)\subseteq \bigcup_{i=1}^jY_i}$ for every $j\in\{1,\ldots,\ell\}$. For $j=1$, then claim holds by (ii$'$). Assume that $j>1$. Then
\begin{equation}\label{eq:ii}
B_j\cap\bigcup_{i=1}^j\cl(S_i)=
(B_j\cap \cl(S_j))\cup((B_j\setminus \cl(S_j))\cap\bigcup_{i=1}^{j-1}\cl(S_i)).
\end{equation}
By (ii$'$), $B_j\cap \cl(S_j)\subseteq Y_j$, and by (ii$'$), 
$B_j\setminus \cl(S_j)=B_{j-1}\setminus \cl(S_j)$. Using \Cref{eq:ii} and the inductive assumption, we obtain that
\begin{align*}
B_j\cap\bigcup_{i=1}^j\cl(S_i)\subseteq &\ Y_j\cup ((B_{j-1}\setminus \cl(S_j))\cap\bigcup_{i=1}^{j-1}\cl(S_i))\subseteq Y_j\cup (B_{j-1}\cap\bigcup_{i=1}^{j-1}\cl(S_i))\\
\subseteq&\ Y_j\cup\bigcup_{i=1}^{j-1}Y_i=\bigcup_{i=1}^jY_j,
\end{align*}
and this proves our claim. Since $\cl(X)=\bigcup_{i=1}^\ell(\cl(S_i))$, $B'=B_\ell$, and $Y=\bigcup_{i=1}^\ell Y_i$
we have that 
\begin{equation*}
B'\cap \cl(X)=B_\ell\cap\bigcup_{i=1}^\ell\cl(S_i)\subseteq\bigcup_{i=1}^\ell Y_i=Y.     
\end{equation*}
Thus, (ii) holds.
This concludes the proof. 
\end{claimproof}

It remains to analyze the size of~$Y$ and the running time.
We next show an upper bound on the size of $Y$ of~$|Y|\leq h^{h^2}\cdot [(k+1)r]^{rh^2}$.
We prove this via induction on $h$.

For $h=1$, note that 
$|Y|\leq (k+1)r\leq h^{h^2}\cdot [(k+1)r]^{rh^2}$. For
$h\geq 2$, consider any independent set $S\subseteq E(\cM)$ of size $h-1$. The algorithm outputs the set $Y_S$, which by the induction hypothesis, has size at most~$N_{h-1}=(h-1)^{(h-1)^2}[(k+1)r]^{r(h-1)^2}$.
By construction, we have
\begin{equation*}
|Y|\leq\max\{(h-1)[(k+1)r]^{r-1}+(k+1)r,\binom{(h-1)[(k+1)r]^{r-1}+(k+1)r-1}{h-1}N_{h-1}\}.  
\end{equation*}
Observe first that
\begin{equation}\label{eq:one}
(h-1)[(k+1)r]^{r-1}+(k+1)r\leq h[(k+1)r]^r  \leq h^{h^2}\cdot [(k+1)r]^{rh^2}.  
\end{equation}
For the second option, note that
\begin{equation*}\label{eq:two}
\begin{aligned}
\binom{(h-1)[(k+1)r]^{r-1}+(k+1)r-1}{h-1}N_{h-1}\leq &\ (h[(k+1)r]^r)^h\cdot N_{h-1}\\
\leq &\ (h[(k+1)r]^r)^h\cdot (h-1)^{(h-1)^2}[(k+1)r]^{r(h-1)^2}\\
\leq&\ h^h[(k+1)r]^{rh}\cdot h^{h^2-h}[(k+1)r]^{r(h^2-h)}\\
\leq&\ h^{h^2}\cdot [(k+1)r]^{rh^2}.
\end{aligned}
\end{equation*}
Combining this equation with \Cref{eq:one}, 
we obtain the required upper bound for $|Y|$.

Finally, we evaluate the running time of $\textsc{Important}$
and show that
$\textsc{Important}(\cM,X)$ runs in $(rk)^{\Oh(rh^2)}\cdot n^{\Oh(1)}$ time.
To this end, notice that $\cl(X)$ can be constructed in polynomial time in the oracle model. 
To construct the $h$-uniform set $Z$ of size $h(rk)^{\Oh(r)}$, we use the greedy procedure where for each $x\in \cl(X)\setminus Z$ for the considered $Z$, we go over all subsets $S$ of $Z$ of size $h-1$ and check whether $S\cup\{x\}$ is independent using the oracle. 
As $h\leq r$, the construction of $Z$ can be done in $(kr)^{\Oh(rh)}\cdot n^{\Oh(1)}$ time. 
The algorithm makes~$(rk)^{\Oh(rh)}$~recursive calls and the depth of the search tree is at most $h$. Summarizing, we obtain that the overall running time is  $(rk)^{\Oh(rh^2)}\cdot n^{\Oh(1)}$.

To complete the proof and construct $W$, we call $\textsc{Important}(\cM,X)$ for an arbitrary basis~$X$ of $\cM$ and set $W=Y$ for the output set of the algorithm. Note that a basis can be found in polynomial time using the independence oracle. Then \Cref{cl:main} implies that if~$\cM$ has a~$k$-fault-tolerant basis, then it also has one in~$W$. Note that the other direction is trivial as~$W \subseteq E(\cM)$.
This concludes the proof.
\end{proof}

We are now ready to prove \Cref{thm:fpt}, which we restate here for convenience.

\fptrankandk*

\begin{proof}
Let $\cM$ be a matroid of rank $r$ and let $k\geq 0$ be an integer. \ftb is trivial for $r=0$ and we can assume that $r\geq 1$. Notice that loops of $\cM$, that is, elements~$e$ such that~$\{e\}\notin\cI(\cM)$ are irrelevant---a loop $e$ is not included in any $k$-fault-tolerant basis and~$e\in\cl(X)$ for any set $X\subseteq E(\cM)$. Hence, we can preprocess $\cM$ and delete the loops. From now on, we assume that $\{e\}\in\cI(\cM)$ for every $e\in E(\cM)$.

We apply the algorithm from~\Cref{lem:important}, and in time $(rk)^{\Oh(r^3)}\cdot n^{\Oh(1)}$, find a set $W\subseteq E(\cM)$ of size at most $r^{r^2}\cdot [(k+1)r]^{r^3}$ such that whenever~$\cM$ has a $k$-fault-tolerant basis, $\cM$ has a $k$-fault-tolerant basis $B\subseteq W$. We consider all candidate subsets $B$ of $W$ with $|B|\leq(k+1)r$ using the upper bound for the size of a $k$-fault-tolerant basis from~\Cref{prop:bounds}. This can be done in $(kr)^{\Oh(kr^4)}$ time. For each candidate set $B$ of size at most $(k+1)r$, we verify whether $B$ is a $k$-fault-tolerant basis as follows. In $\Oh([(k+1)r]^{k})$ time, we check whether~$\rank(B\setminus F)=r$ for every subset $F\subseteq B$ of size at most $k$. Among all candidate sets satisfying the above property, we select a set of minimum size which is a $k$-fault-tolerant basis of $\cM$. The overall running time is $(kr)^{\Oh(kr^4)}\cdot n^{\Oh(1)}$. This concludes the proof.
\end{proof}

\section{Partition matroids}\label{sec:partition}
In this section, we prove \Cref{thm:poly-part}. The proof is based on the following structural lemma.

\begin{lemma}\label{lem:partition}
Let $\cM$ be a partition matroid with unit capacities for a partition $(P_1,\ldots,P_d)$ of the ground set.
Let~$X\subseteq E(\cM)$ and let $k\geq 0$ and $r\geq 1$ be integers. Then, the following properties hold.
\begin{itemize}
\item[(i)] If there is an integer $s\geq 1$ such that $|X|=s(r-1)+k+1$ and $|X\cap P_i|\leq s$ for every $i\in\{1,\ldots,d\}$, then 
for any set $F\subseteq X$ of size at most $k$, $\rank(X\setminus F)\geq r$.
\item[(ii)] If $X$ is an inclusion-minimal set 
such that~${\rank(X\setminus F)\geq r}$ for any set $F\subseteq X$ of size at most $k$,
then there is an integer $s\geq 1$ such that $|X|=s(r-1)+k+1$ and~$|X\cap P_i|\leq s$ for every $i\in\{1,\ldots,d\}$.
\end{itemize}
\end{lemma}

\begin{proof}
To show (i), 
suppose that $|X|=s(r-1)+k+1$ and $|X\cap P_i|\leq s$ for each~${i\in\{1,\ldots,d\}}$ for some integer $s\geq 1$. 
Consider any subset $F\subseteq X$ of size at most $k$. Since~${|X\cap P_i|\leq s}$ for every~$i\in\{1,\ldots,d\}$, it holds that~$|(X \setminus F) \cap P_i| \leq s$ and hence~$|(X \setminus F) \cap P_i| > 0$ for at least~$\Big\lceil\frac{|X \setminus F|}{s}\Big\rceil$ sets~$P_i$.
Since~$\cM$ is a partition matroid with unit capacities, this is also a lower bound on~$\rank(X \setminus F)$ and therefore
\begin{equation*}
\rank(X\setminus F)\geq \Big\lceil\frac{|X\setminus F|}{s}\Big\rceil\geq 
\Big\lceil\frac{s(r-1)+1}{s}\Big\rceil\geq r.
\end{equation*}
This proves that $\rank(X)\geq r$ and~$\rank(X\setminus F)\geq r$ for any set $F\subseteq X$ of size at most $k$.

To prove (ii), suppose that $X$ is an inclusion-minimal set with $\rank(X)\geq r$ such that for any set $F\subseteq X$ of size at most $k$, $\rank(X\setminus F)\geq r$ holds. We assume without loss of generality that~$|X\cap P_1|\geq\dots\geq|X\cap P_d|$. 

If $r=1$, then we set $s=\max\{|X\cap P_i|\mid 1\leq i\leq d\}$. Trivially, $|X|\geq k+1$ and~$|X\cap P_i|\leq s$ for every $i\in\{1,\ldots,d\}$. 
We next show that $|X|=k+1$. Assume towards a contradiction that~$|X| > k+1$.
Any subset $X'\subseteq X$ of size $k+1$ satisfies~$|X'\cap P_i|\leq s$ for every~$i \in \{1,\ldots,d\}$.
Hence, $\rank(X' \setminus F) > 1$ for any set~$F$ of size at most~$k$ by (i). This contradicts the the minimality of $X$ and shows~$|X|=k+1$.

If $r\geq 2$, then $|X\cap P_{r-1}|\geq 1$ 
and we set $s=|X\cap P_{r-1}|$. We show~${|X|=s(r-1)+k+1}$ and~$|X\cap P_i|\leq s$ for every $i\in\{1,\ldots,d\}$ next.
Let ${Y=\bigcup_{j=1}^{r-1}(X\cap P_j)}$ and $Z=\bigcup_{j=r}^d(X\cap P_j)$. 
By definition of $s$, $|Y|\geq s(r-1)$.

If~$|Z|\leq k$, then ${\rank(X\setminus Z)=r-1<r}$, 
contradicting the choice of $X$. 
Thus, $|Z|\geq k+1$ and 
$|X|=|Y|+|Z|\geq s(r-1)+k+1$. 
There exist $Y'\subseteq Y$ such that 
$|Y'\cap P_j|=s$ for each~$j\in\{1,\ldots,r-1\}$, and $Z'\subseteq Z$ such that $|Z'|=k+1$. Consider $X'=Y'\cup Z'$. By definition, 
$|X'\cap P_i|\leq s$ for every $i\in\{1,\ldots,d\}$. Then, by (i), we have that $\rank(X')\geq r$ and for any set $F\subseteq X'$ of size at most $k$, $\rank(X'\setminus F)\geq r$. By the minimality of $X$, we conclude that $X=X'$. This proves~(ii). 
\end{proof}

We are now ready to prove \Cref{thm:poly-part}, which we restate here.

\partition*

\begin{proof}
Consider a weighted partition matroid $\cM$ with unit capacities and a weight function~$w\colon E(\cM)\rightarrow \mathbb{Z}_{\geq 0}$. 
Let $k\geq 0$ and $r\geq 0$ be integers. 
Let also $(P_1,\ldots,P_d)$ be the partition of~$E(\cM)$ defining $\cM$. 
If $r=0$, then the claim of the theorem is trivial as $X=\emptyset$ is a solution. Thus, we can assume that $r\geq 1$. 

For an integer $s\geq 1$, we say that $s$ is \emph{feasible} if a set~$X\subseteq E(\cM)$ of size~${s(r-1)+k+1}$ exists such that $|X\cap P_i|\leq s$ for every $i\in \{1,\ldots,d\}$.
For each integer~${1\leq s\leq \max\{|P_i|\mid 1\leq i\leq d\}}$, we check whether~$s$ is feasible by checking whether~${\sum_{i=1}^d|\min\{|P_i|,s\}|\geq s(r-1)+k+1}$. 
If there is no feasible integer, then we conclude that there is no 
$X\subseteq E(\cM)$ with the property that~$\rank(X\setminus F)\geq r$ for any subset $F\subseteq X$ of size at most $k$ by \Cref{lem:partition}. Otherwise, for every feasible $s$, we greedily choose a set $X$ of minimum weight such that~${|X|=s(r-1)+k+1}$ and~$|X\cap P_i|\leq s$ for every~$i\in \{1,\ldots,d\}$. Notice that the selection of such a set is equivalent to finding a minimum weight independent set of size $s(r-1)+k+1$ in the partition matroid for~$(P_1,\ldots,P_d)$ with the capacities~$(c_1,\ldots,c_d)$ where $c_i=\min\{s,|P_i|\}$ for each~$i\in\{1,\ldots,d\}$. By the well-known matroid properties~\cite{Oxley92}, the greedy algorithm finds such a set $X$. By \Cref{lem:partition}, we conclude that $X$ is a minimum weight set 
such that for any set~$F\subseteq X$ of size at most $k$, $\rank(X\setminus F)\geq r$. This concludes the description of our algorithm and its correctness proof.

To evaluate the running time, notice that the elements of $\cM$ can be sorted by their weight in~$\Oh(n\log n)$ time. Then we have at most $n$ choices of $s$ and for each $s$, the greedy algorithm works in $\Oh(n)$ time. Thus, the overall running time is in~$\Oh(n^2)$, concluding the proof.
\end{proof}

As a corollary, we obtain the following for general matroids of rank at most two.

\begin{restatable}{corollary}{polyrank}
    \label{cor:poly-rank}
    \wftb can be solved in $\Oh(n^2)$ on matroids of rank at most two given by an independence oracle.
\end{restatable}

\begin{proof}
Let $(\cM,w,k)$ be an instance of \wftb with~$\rank(\cM)\leq 2$. If $\rank(\cM)=0$, then the problem is trivial as $B=\emptyset$ is a $k$-fault-tolerant basis. Thus, we assume that $\rank(\cM)\geq 1$. 
Furthermore, we assume that $\cM$ has no loops because elements $e\in E(\cM)$ such that $\{e\}\notin \cI(\cM)$ are not included in any $k$-fault-tolerant basis and may be removed in linear time. Hence, $\rank(\{e\})\geq 1$ for any $e\in E(\cM)$.

Suppose that $\rank(\cM)=1$. Then, $k$-fault-tolerant bases are exactly the sets of size $k+1$ and we can choose  such a set of minimum weight greedily.

We assume from now on that $\rank(\cM)=2$. For two $x,y\in E(\cM)$, we say that~$x$ and~$y$ are \emph{collinear} if $y\in \cl(\{x\})$ (or, symmetrically, $x\in\cl(\{y\})$). Notice that this is an equivalence relation. It is trivial that the relation is reflexive, and (CL1)--(CL4) imply transitivity. To see that the relation is symmetric, assume that $y\in\cl(\{x\})$. Then, by definition, ${\rank(\{x,y\})=\rank(\{x\})=1}$. 
If $x\notin \cl(\{y\})$, then $\rank(\{x,y\})>\rank({y})=1$ immediately leading to a contradiction. We partition $E(\cM)$ into equivalence classes $(P_1,\ldots,P_d)$ of collinear elements. Observe that a set $X\subseteq E(\cM)$ is of rank two if and only if $X$ contains elements of at least two sets in the partition. 
We consider the partition matroid $\cM'$ with unit capacities defined by~$(P_1,\ldots,P_d)$ and obtain that $B\subseteq E(\cM)$ is a solution to the instance $(\cM,w,k)$ if and only if $B$ is a minimum weight set 
such that for any set $F\subseteq B$ of size at most $k$, 
the rank of $B\setminus F$ with respect to $\cM'$ is at least two. This allows us to apply \Cref{thm:poly-part} to solve the problem. 

Since the problem for the case $r=1$ can be solved in $\Oh(n\log n)$ time and the partition~$(P_1,\ldots,P_d)$ can be constructed in linear time, we obtain that the overall running time is in~$\Oh(n^2)$. This concludes the proof.
\end{proof}

\section{Complexity dichotomy for the parameterization by rank }\label{sec:lower}
In this section, we prove \Cref{thm:hard-r} which we restate here for convenience.

\hardness*

\begin{proof}
The claim for $r\leq 2$ is proven in \Cref{cor:poly-rank}. Thus, it remains to show the computational lower bound for $r\geq 3$. 
We show the claim for $r=3$ and then explain how to extend it for any~$r\geq 3$.
We reduce from \textsc{General Position Subset Selection}. Recall that a set~$P$ of points on the Euclidean plane is said to be in general position if there are no three points on the same line on the plane. \textsc{General Position Subset Selection} is defined as follows.
Given a set~$P$ of point on the plane and an integer $p\geq 1$, decide whether there is a subset~$Q\subseteq P$ of at least~$p$ points in general position. 
The problem was shown to be \classNP-hard by Froese et al.~\cite{FroeseKNN17} and the result holds for points with rational coordinates.

Let $(P=\Big\{\begin{pmatrix} x_1 \\ y_1 \end{pmatrix},\ldots,
\begin{pmatrix} x_n \\ y_n \end{pmatrix}
\Big\},p)$ be an instance of \textsc{General Position Subset Selection}. We assume without loss of generality that~$p\geq 3$ as any set of at most two points is in general position. We set $k=p-3$, construct the set~$E=\Big\{\begin{pmatrix} x_1 \\ y_1\\ 1 \end{pmatrix},\ldots,
\begin{pmatrix} x_n \\ y_n \\1 \end{pmatrix}
\Big\}$ of vectors, and define $\cM$ to be the linear matroid with the ground set $E$ over rationals. 

Note that~$\rank(\cM) \leq 3$ and that three distinct points $\begin{pmatrix} x_h \\ y_h \end{pmatrix}$, $\begin{pmatrix} x_i \\ y_i \end{pmatrix}$, and $\begin{pmatrix} x_j \\ y_j \end{pmatrix}$ of $P$ are not on the same line if and only if 
$\rank\Big(\Big\{\begin{pmatrix} x_h \\ y_h\\1 \end{pmatrix},\begin{pmatrix} x_i \\ y_i\\1 \end{pmatrix},\begin{pmatrix} x_j \\ y_j\\1 \end{pmatrix}\Big\}\Big)=3$. This implies that for any~$I\subseteq \{1,\ldots,n\}$,
the points of 
$\Big\{\begin{pmatrix} x_i \\ y_i\end{pmatrix}\mid i\in I
\Big\}\subseteq P$ are in general position if and only if
$\Big\{\begin{pmatrix} x_i \\ y_i\\ 1\end{pmatrix}\mid i\in I
\Big\}\subseteq E$ is a $3$-uniform set for $\cM$. 
Hence, $P$ has a subset of at least $p$ points in general position if and only if there is a $3$-uniform subset of $E$ of size at least $p$. By \Cref{prop:bounds}, a $k$-fault-tolerant basis of $\cM$ has size at least $p=k+3$.
By \Cref{obs:minbase}, $\cM$ has a $k$-fault-tolerant basis of size $p$ if and only if $(P,p)$ is a yes-instance of  \textsc{General Position Subset Selection}. This concludes the proof for $r=3$.

To see the claim for $r>3$, we perform the same reduction but in~$r$ dimensions. The constructed vectors have a 0 in all but the first three dimensions (where they have the entries as constructed above).
For each~$4 \leq i \leq r$, 
we then add $k+1$ vectors with zero-entries everywhere except for dimension~$i$, where the entry is~$1$.
Note that any $k$-fault-tolerant basis of the linear matroid over the constructed vectors has to contain all~$(r-3)(k+1)$ newly added vectors.
Moreover, it has to contain a set~$P$ of vectors such that after removing any~$k$ of them, the first three dimensions have to be spanned by the remaining vectors.
This is equivalent to the case~$r=3$ and concludes the proof.
\end{proof}

\section{Conclusion}\label{sec:concl}
In this paper, we initiate the study of the parameterized complexity of \ftb. Our main result is that the problem is fixed-parameter tractable when parameterized by both~$k$ and rank~$r$ of the input matroid. This positive algorithmic result is complemented by computational lower bounds showing \classNP-hardness for constant values for any one of the two parameters alone. 
Our results lead to several open questions.

First, we do not know whether our \classFPT result from~\Cref{thm:fpt} could be extended for \wftb. Our approach based on the choice of $h$-uniform sets is tailored for the unweighted case, and one may need different techniques in the presence of weights.

Second, is it possible to extend our result to the non-uniform model introduced by Adjiashvili et al.~\cite{AHMS22}? Here, we assume that the set of elements of a matroid $\cM$ is partitioned into two subsets $S$ and $V$ of \emph{safe} and \emph{vulnerable} elements, respectively. Then, the task is to either find a set $B\subseteq E(\cM)$ of minimum size such that for any sets $F\subseteq V\cap B$ of size at most $k$, $\rank(\cM-F)=\rank(\cM)$ or correctly report that such a set does not exist. \ftb is the special case of this problem with $S=\emptyset$.

Third, can our results be extended to the model introduced by Adjiashvili et al.~\cite{Adjiashvili2015bulk} where not arbitrary sets of~$k$ elements can fail but possible failure scenarios are part of the input?

Finally, observe that our computational lower bounds from \Cref{prop:hardtodecide}, \Cref{obs:hard-k}, and \Cref{thm:hard-r} do not exclude efficient algorithms for \ftb on special classes of matroids. In particular, we proved in \Cref{thm:poly-part} that \wftb can be solved in polynomial time on truncations of partition matroids with unit capacities. 
While it is straightforward to see that \wftb can be solved in polynomial time for partition matroids with arbitrary capacities,\footnote{Given a partition matroid defined by a partition $(P_1,\ldots,P_d)$ of the ground set and a $d$-tuple of capacities~$(c_1,\ldots,c_d)$, to solve \wftb, we have to choose $c_i+k$ elements of minimum weight from each $P_i$.
}
it is not clear whether the problem can be solved in polynomial time on their truncations. The problem complexity for other fundamental classes of matroids, like transversal matroids and gammoids, is another interesting open question.


\end{document}